\documentclass[conference,a4paper]{IEEEtran}

\usepackage{color}
\usepackage{amssymb} 
\usepackage{mathtools}
\usepackage{footnote}
\usepackage{amsfonts}
\newcommand{\ve}[1]{{\bf #1}}
\newcommand{\mode}{\mbox{mod }}
\newcommand{\vol}{\mbox{Vol }}
\newcommand{\norm}[1]{\left|\left|#1\right|\right|}
\newcommand{\prob}[1]{\text{Pr}\left(#1\right)}

\newtheorem{lemma}{Lemma}
\newtheorem{theorem}{Theorem}

\newtheorem{remark}{Remark}

\begin{document}

\sloppy
\IEEEoverridecommandlockouts
\title{Asymmetric Compute-and-Forward with CSIT}

 \author{
   \IEEEauthorblockN{Jingge Zhu and Michael Gastpar}
   \IEEEauthorblockA{School of Computer and Communication Sciences, EPFL\\
   Lausanne, Switzerland\\
     Email: \{jingge.zhu, michael.gastpar\}@epfl.ch} 
 \thanks{This work was supported in part by the European ERC
Starting Grant 259530-ComCom.}}
\maketitle

\begin{abstract}
We present a modified compute-and-forward scheme which utilizes Channel State Information at the Transmitters (CSIT) in a natural way. The modified scheme allows different users to have different coding rates, and use CSIT to achieve larger rate region. This idea is applicable to all systems which use the compute-and-forward technique and can be arbitrarily better than the regular scheme in some settings.
\end{abstract}

\section{Introduction}
The compute-and-forward scheme \cite{NazerGastpar_2011} is a novel coding scheme for Gaussian networks which takes advantage of the linear structure of lattice codes and the additive nature of  Gaussian interference networks. The main idea of compute-and-forward is to decode linear combinations of messages rather than the messages themselves at the receivers. One nice feature of this scheme is that the Channel State Information is not explicitly required at the transmitters, making this scheme attractive to practical considerations. But on the other hand, it is not clear how CSI can be used at the transmitters if the compute-and-forward idea is to be applied. Furthermore, the same lattice code is used for every user, preventing the scheme from exploiting  asymmetries of the networks.

In this work we present a modified compute-and-forward scheme with asymmetric message rates which makes CSIT useful. This scheme also extends the concept of the computation rate to a more general definition of the computation rate tuple, which allows flexibility in controlling the individual message rates of different users.  The main idea relies on the observation, roughly speaking, that the transmitted lattice codeword does not have to lie in the lattice which is used for lattice coding at the transmitters.

We use the notation $[a:b]$ to denote a set of increasing integers $\{a,a+1,\ldots,b\}$, $\log$ to denote $\log_2$ and $\log^+(x)$to denote the function $\max\{\log(x),0\}$.  We also use $x_{1:K}$ to denote a set of numbers $\{x_1,x_2,\ldots,x_K\}$.

\section{problem statement }
We consider a interference network with $K$ transmitters and $M$ relays. The discrete-time real Gaussian channel has the following vector representation
\begin{IEEEeqnarray*}{rCl}
\ve y_m=\sum_{k=1}^Kh_{mk}\ve x_k+\ve z_m,\quad m\in[1:M]
\end{IEEEeqnarray*}
with $\ve y_m\in\mathbb R^n, \ve x_k\in\mathbb R^n, h_{mk}\in\mathbb R$ denoting the channel output of relay $m$, channel input of transmitter $k$ and the channel gain, respectively. The Gaussian white noise with unit variance is denoted by $\ve z_m\in\mathbb R^n$. We impose the same  power constraint $\mathbb E\{\norm{\ve x_k}^2\}\leq nP$ on all the transmitters.
 
The message of user $k$ is represented by a point in $\mathbb R^n$ denoted by $\ve t_k$, which is an element of the codebook $\mathcal C_k$ of user $k$ with \textit{message rate} $R_k:=\frac{1}{n}\log|\mathcal C_k|$.

Each transmitter is equipped with an encoder $\mathcal E_k$ which maps its message into the channel input as $\ve x_k=\mathcal E_k(\ve t_k)$.  Each relay $m$ has a decoder $\mathcal D_m$ which uses the channel output $\ve y_m$ to decode a function of all the messages $\ve t_k, k\in[1:K]$ as $f_m(\ve t_{[1:K]})=\mathcal D_m(\ve y_m)$.  Here we only consider the function of the form 
$f_m(\ve t_{[1:K]})=\sum_{k=1}^Ka_{mk}\ve t_k$
with  integer $a_{mk}$ for all $m\in[1:M],k\in[1:K]$.  We use $\ve a_m$ to denote the column vector $[a_{m1},\ldots,a_{mK}]^T$. 
 
We say a \textit{computation rate tuple} $(R_1,\ldots,R_K)$ \textit{with respect to the function} $f_m$ is achievable, if the relay $m$ can decode the function $f_m$ reliably, namely $ \prob{\mathcal D_m(\ve y_m)\neq f_m(\ve t_{[1:K]})}<\delta$ for any $\delta>0$, with $R_k$ being the message rate of the user $k$. In the network, we require all the relays to decode their desired functions. We say \textit{a computation rate tuple $(R_1,\ldots, R_K) $ w. r. t. the set of functions $f_m,m\in[1:M]$} is achievable, if
$\prob{\mathcal D_m(\ve y_m)\neq f_m(\ve t_{[1:K]}), \text{ for all }m\in[1:M] }<\delta$
holds for any $\delta>0$ with $R_k$ being the message rate of user $k$. In the following we will study the computation rate tuple achieved by a modified compute-and-forward scheme.

\section{Lattice codes construction}
A lattice $\Lambda$ is a discrete subgroup of $\mathbb R^n$ with the property that if $\ve t_1, \ve t_2\in \Lambda$, then $\ve t_1+\ve t_2\in \Lambda$.  More details about lattice and lattice codes can be found in \cite{ErezZamir_2004}. Define the lattice quantizer $Q_{\Lambda}: \mathbb R^n\rightarrow\Lambda$ as
$Q_{\Lambda}(\ve x)=\mbox{argmin}_{\ve t\in\Lambda}\norm{\ve t-\ve x}$
and define the fundamental Voronoi region of the lattice to be
$\mathcal V:=\{\ve x\in\mathbb R^n:Q_{\Lambda}(\ve x)=\ve 0\}$.
The modulo operation gives the quantization error:
$[\ve x]\mode\Lambda=\ve x-Q_{\Lambda}(\ve x)$.
Two lattices $\Lambda$ and $\Lambda'$ are said to be nested if $\Lambda'\subseteq\Lambda$. 

Let $\Lambda_1,\ldots,\Lambda_M$ be $M$ nested lattice codes constituting a nested lattice chain in which all lattices  are simultaneously good in the sense of \cite{ErezZamir_2004}.  This chain can be constructed as shown in \cite{Nam_etal_2010} and the order of the chain will be determined later. Relay $m$ will perform the lattice decoding with respect to the lattice $\Lambda_m$.

Let $\beta_1,\ldots,\beta_K$ be $K$ positive numbers. We can construct $K$ nested lattices  such that $\Lambda_k^{s}\subseteq\Lambda_{c}$ for all $k$ where $\Lambda_{c}$ denotes the coarsest lattice among $\Lambda_1,\ldots,\Lambda_M$.  We  let $\Lambda_k^{s}$  to be simultaneously good and with second moment 
$\frac{1}{n\vol(\mathcal V_k^s)}\int_{\mathcal V_k^{s}}\norm{\ve x}^2d\ve x=\beta_k^2 P$
where $\mathcal V_k^{s}$ denotes the Voronoi region of the lattice $\Lambda_k^{s}$ for $k\in[1:K]$.  The lattice $\Lambda_k^s$ is used as the shaping region for the codebook of user $k$. 

For each transmitter $k$, we construct the codebook
\begin{IEEEeqnarray}{rCl}
\mathcal C_k&=&\Lambda_{m(k)}\cap\mathcal V_k^s
\label{eq:codebook}
\end{IEEEeqnarray}
where $m(k)\in\{1,\ldots,M\}$ hence $\Lambda_{m(k)}$ is the decoding lattice at one of the $M$ relays. We will determine which decoding lattice to choose for transmitter $k$, i.e., the expression of $m(k)$,  in the next section.  The message rate of user $k$ is
\begin{IEEEeqnarray}{rCl}
R_k=\frac{1}{n}\log|\mathcal C_k|=\frac{1}{n}\log\frac{\vol(\mathcal V_k^s)}{\vol(\mathcal V_{m(k)})}
\label{eq:message_rate}
\end{IEEEeqnarray}
where $\mathcal V_{m(k)}$ is the Voronoi region of the fine lattice $\Lambda_{m(k)}$.

\section{A modified compute-and-forward scheme}
When the message (codeword) $\ve t_k$ is given to encoder $k$, it forms its channel input as follows
\begin{align*}
\ve x_k=\left[\ve t_k/\beta_k+\ve d_k\right]\mode\Lambda_k^s/\beta_k
\end{align*}
where $\ve d_k$ is called a \textit{dither} which is a random vector uniformly distributed in the scaled Voronoi region $\mathcal V_k^s/\beta_k$. As pointed out in \cite{ErezZamir_2004}, $\ve x_k$ is independent from $\ve t_k$ and also uniformly in $\Lambda_k^s/\beta_k$ hence  has average power $P$ for all $k$.

To demonstrate the proposed approach, we first assume there is only one relay, $m=M=1$. For now there is only one decoding lattice hence the codebooks of all the users are constructed using the same fine lattice and we denote it as $\Lambda_{m(k)}=\Lambda$ for all $k$.
\begin{theorem}
Assume there is only one relay $m$. For any given set of positive numbers $\beta_1,\ldots,\beta_K$, there exists lattice codes $\mathcal C_1,\ldots,\mathcal C_K$ such that the achievable computation rate tuple $(R_1,\ldots,R_K)$  with respect to the function $f_m=\sum_ka_{mk}\ve t_k$ at  relay $m$ is given by
\begin{IEEEeqnarray}{rCl}
&R_k&<r_k(\ve h_m,\ve a_m,\beta_{1:K})\nonumber\\
&:=&\left[\frac{1}{2}\log \left(\norm{\ve{\tilde a}_m}^2-\frac{P(\ve h_m^T\ve{\tilde a}_m)^2}{1+P\norm{\ve h_m}^2}\right)^{-1}+\frac{1}{2}\log \beta_k^2\right]^+
\label{eq:compute_rate}
\end{IEEEeqnarray}
for all $k$ with $\ve{\tilde a}_m:=[\beta_1a_{m1},...,\beta_Ka_{mK}]$ and $a_{mk}\in\mathbb Z$  for all $k\in[1:K]$.
\label{thm:R_comp_OneFun}
\end{theorem}

\begin{proof} At the decoder we form
\begin{IEEEeqnarray*}{rCl}
\tilde {\ve y}_m&:=&\alpha_m\ve y_m-\sum_ka_{mk}\beta_k\ve d_k\\
&=&\sum_k a_{mk}\left(\beta_k (\ve t_k/\beta_k+\ve d_k)-\beta_k Q_{\Lambda_k^s/\beta_k}(\ve t_k/\beta_k+\ve d_k)\right)\\
&&-\sum_ka_{mk}\beta_k\ve d_k+\tilde{\ve  z}_m\\
&\stackrel{(a)}{=}&\tilde {\ve z}_m+\sum_k a_{mk} (\ve t_k-Q_{\Lambda_k^s}(\ve t_k+\beta_k\ve d_k))\\
&:=&\tilde {\ve z}_m+\sum_k a_{mk} \tilde{\ve t}_k
\end{IEEEeqnarray*}
with $\tilde {\ve t}_k:=\ve t_k-Q_{\Lambda_k^s}(\ve t_k+\beta_k\ve d_k)$ and the equivalent noise 
\begin{IEEEeqnarray*}{rCl}
\tilde {\ve z}_m:=\sum_k(\alpha_m h_{mk}-a_{mk}\beta_k)\ve x_k+\alpha_m\ve z_m
\end{IEEEeqnarray*}
which is independent of $\sum_k a_{mk}\tilde{\ve t}_k$ since all $\ve x_k$ are independent of $\sum_k a_{mk}\tilde {\ve t}_k$ thanks to the dithers $\ve d_k$. The step $(a)$ follows because it holds $Q_{\Lambda}(\beta  X)=\beta Q_{\frac{\Lambda}{\beta}}(X)$ for any $\beta\neq 0$. 
Notice we have $\tilde {\ve t}_k\in\Lambda$ since  $\ve t_k\in\Lambda$ and $\Lambda_k^s\subseteq\Lambda$ due to the code construction.  Hence the linear combination $\sum_k a_{mk}\tilde{\ve t}_k$ along belongs to the decoding lattice $\Lambda$.

The relay uses lattice decoding to decode $\sum_ka_{mk}\tilde{\ve t}_k$ with respect to the decoding lattice $\Lambda$ by quantizing $\tilde {\ve y}_m$ to its nearest neighbor in  $\Lambda$. The decoding error probability  is equal to the probability that the equivalent noise $\tilde {\ve z}_m$ leaves the Voronoi region surrounding the lattice point representing $\sum_ka_{mk}\tilde{\ve t}_k$. If the fine lattice $\Lambda$ used for decoding is good for AWGN channel, as it is shown in \cite{ErezZamir_2004}, the probability $\prob{\tilde {\ve z}_m\notin \mathcal V}$ goes to zero exponentially if 
\begin{IEEEeqnarray}{rCl}
\frac{\vol(\mathcal V)^{2/n}}{N_m}> 2\pi e
\label{eq:volume_to_noise}
\end{IEEEeqnarray}
where $N_m:=\mathbb E\norm{\tilde {\ve z}_m}^2/n=\norm{\alpha_m\ve h-\tilde{\ve a}_m}^2P+\alpha_m^2$ denotes the average power per dimension of the equivalent noise. Recall that the shaping lattice $\Lambda_k^s$ is good for quantization hence we have 
\begin{IEEEeqnarray}{rCl}
\vol(\mathcal V_k^s)=\left(\frac{\beta_k^2P}{G(\Lambda_k^s)}\right)^{n/2}
\end{IEEEeqnarray}
with $G(\Lambda_k^s)2\pi e<(1+\delta)$ for any $\delta>0$ if $n$ is large enough \cite{ErezZamir_2004}. Together with the message rate expression in (\ref{eq:message_rate}) (here $\Lambda_{m(k)}=\Lambda$ for all $k$) we can see that lattice decoding is successful if $\beta_k^2P2^{-2R_k}/G(\Lambda_k^s)>2\pi e N_m$ for every $k$
or equivalently
\begin{IEEEeqnarray*}{rCl}
R_k<\frac{1}{2}\log\left(\frac{P}{N_m}\right)+\frac{1}{2}\log\beta_k^2-\frac{1}{2}\log(1+\delta)
\end{IEEEeqnarray*}
By choosing $\delta$ arbitrarily small and optimizing over $\alpha_m$ we conclude that under the  rate constraints in (\ref{eq:compute_rate}) the lattice decoding of $\sum_ka_k\tilde{\ve t}_k$ will be successful.  Finally, since there is a one-to-one mapping between $\tilde {\ve t}_k$ and $\ve t_k$ when the dithers $\ve d_k$ are known,  we can also recover $\sum_ka_k\ve t_k$.  It is easy to see from the expression of the computation rate tuple in (\ref{eq:compute_rate}), that  multiplying all $\beta_k$ with a same constant will not change the result.
\end{proof}

We see the main difference to the regular compute-and-forward scheme is that here the transmitted signal $\ve x_k$ contains the scaled version, $\ve t_k/\beta_k$,  of the codeword while the receivers still perform the lattice decoding w. r. t. the lattice $\Lambda$ in which $\ve t_k$ lies. We should choose $\beta_k$  appropriately according to the function $\ve a$ and the channel $\ve h$ to obtain the best rate region.

Now we extend the result to allow all relays to be able to decode their desired linear functions. 
\begin{theorem}
For any given set of positive numbers $\beta_1,\ldots,\beta_K$, there exist lattice codes $\mathcal C_1,\ldots,\mathcal C_K$ such that the achievable computation rate tuple $(R_1,\ldots,R_K)$  with respect to the set of functions $f_m=\sum_ka_{mk}\ve t_k$, $m\in[1:M]$,  where $f_m$ is desired by relay $m$ with $a_{mk}\in\mathbb Z$, is given by
\begin{IEEEeqnarray*}{rCl}
R_k&<&\min_{m\in[1:M]}r_k(\ve h_m,\ve a_m,\beta_{1:K})
\end{IEEEeqnarray*}
where $r_k(\ve h_m,\ve a_m,\beta_{1:K}) $ is defined in (\ref{eq:compute_rate}).
\end{theorem}
\begin{proof}
Relay $m$  decodes the function $f_m$ with its decoding lattice $\Lambda_m$. The nested structure of fine lattices $\Lambda_m$ ensures that the sum of codewords seen at relay $m$ lies in the decoding lattice $\Lambda_m$. As in Theorem \ref{thm:R_comp_OneFun}, the lattice decoding is successful if the volume-to-noise ratio of the decoding lattice satisfies equation (\ref{eq:volume_to_noise}).  Hence each relay imposes a constraint on the individual message rate, i.e., for all $k$, we have $R_k\leq r_k(\ve h_m,\ve a_m,\beta_{1:K})$ for all $m$. If all relays want to decode successfully, each transmitter has to construct its codebook such that it meets the above constraints at all relays. Therefore when the codebook is constructed as in (\ref{eq:codebook}), the fine lattice $\Lambda_{m(k)}$ for $\mathcal C_k$ should be such that the message rate $R_k$ does not exceed $\min_{m\in[1:M]}r_k(\ve h_m,\ve a_m,\beta_{1:K})$, i.e., 
$m(k)=\text{argmin}_{m\in[1:M]}r_k(\ve h_m,\ve a_m,\beta_{1:K})$.  The noise variance $N_m$ at each relay determines the order of the lattice chain involving $\Lambda_m$: larger $N_m$ corresponds to a coarser lattice.
\end{proof}

\begin{remark}
The original scheme in \cite{NazerGastpar_2011} is a special case of this modified scheme with $\beta_k=1$ for all $k$.  In \cite{NazerGastpar_2011}, all message rates are forced to be the same, called \textit{computation rate} at this relay. The modified scheme  allows for different message rates among users and leads to the more general definition \textit{computation rate tuple}. We shall see that this asymmetry on message rates can be beneficial in various scenarios.
\end{remark}

\begin{remark}
The modified scheme extends naturally to the case when transmitters have different power constraints, and in general achieves larger computation rate region.
\end{remark}

\section{Examples}

\textbf{Example 1: The multiple access channel (MAC).}

We consider a $2$-user Gaussian MAC where the receiver wants to decode a linear function of the two messages. Figure \ref{fig:2UserMAC} shows the achievable rate regions.

\begin{figure}[hb!]
\centering
\includegraphics[scale=0.48]{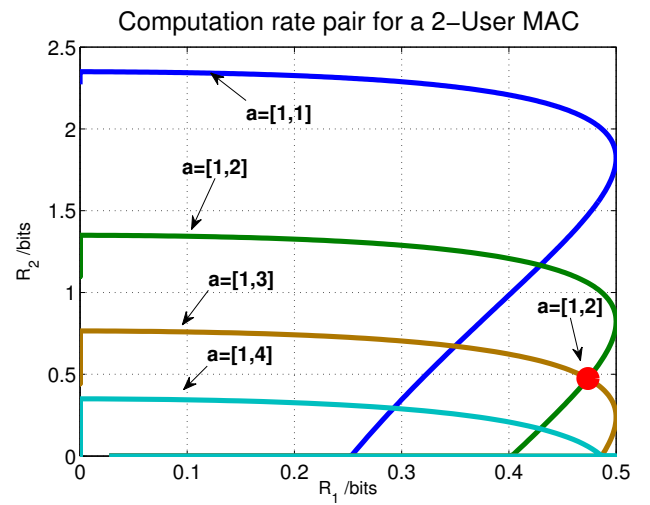}
\caption{We consider a 2-user Gaussian MAC with channel coefficients $\ve h=[1,5]$ and power $P=1$ where the receiver decodes one linear function. Here we show four achievable computational rate pair regions $(R_1,R_2)$ of four different linear functions marked in the plot. For each function, by adjusting parameters $\beta_1$ and $\beta_2$ we can achieve different points on the curve. The red dot indicates the equal rate pair achieved with the best coefficients ($\ve a=[1,2]$ in this case) using the regular compute-and-forward, given in \cite[Thm. 4]{NazerGastpar_2011}.}
\label{fig:2UserMAC}
\end{figure}

\textbf{Example 2: Transmitters with different powers.}

We consider the Gaussian two-way relay channel shown in Figure \ref{fig:system_TWR}, which is studied in \cite{Nam_etal_2010}, \cite{wilson_joint_2010}. Two encoders have different power constraints $P_1$ and $P_2$ and the channel gain from both transmitters is $1$. The relay has power constraint $P_R$. All noises are Gaussian with unit variance.

\begin{figure}[hb!]
\centering
\includegraphics[scale=0.5]{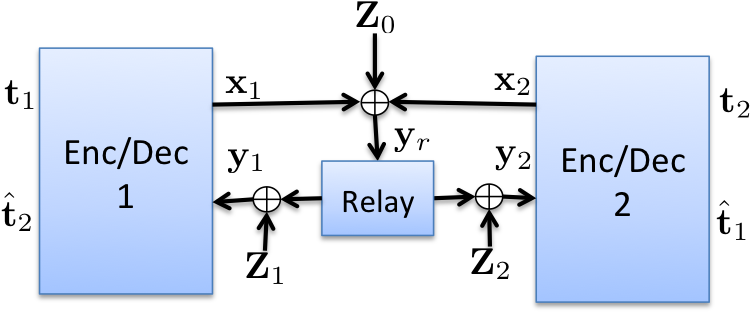}
\caption{A Gaussian two-way relay channel.}
\label{fig:system_TWR}
\end{figure}

Already shown in \cite{Nam_etal_2010}, \cite{wilson_joint_2010}, it can be beneficial for the relay to decode a linear combination of the two messages rather than decoding the two messages individually. They give the  following achievable rate for this network
\begin{subequations}
\begin{align}
R_1\leq &\min\left\{\frac{1}{2}\log^+\left(\frac{P_1}{P_1+P_2}+P_1\right),\frac{1}{2}\log(1+P_R)\right\}\\
R_2\leq &\min\left\{\frac{1}{2}\log^+\left(\frac{P_2}{P_1+P_2}+P_2\right),\frac{1}{2}\log(1+P_R)\right\}
\end{align}
\label{eq:TWR_regular}
\end{subequations}
where the relay decodes the function $\ve t_1+\ve t_2$ and broadcasts it to  two decoders. With the modified compute-and-forward scheme we also ask the relay to decode a linear combination of the form $\sum_{k=1}^2 a_k\ve t_k$ where $a_1,a_2\neq 0$, with which each decoder can solve for the desired message. We can show the following achievable rate region:
\begin{IEEEeqnarray*}{rCl}
R_1&\leq &\min\left\{\frac{1}{2}\log^+\left(\frac{P_1\beta_1^2}{\tilde N(\beta_{1:2})}\right),\frac{1}{2}\log \frac{\beta_1^2P_1(1+P_R)}{P_R}\}\right\}\\
R_2&\leq &\min\left\{\frac{1}{2}\log^+\left(\frac{P_2\beta_2^2}{\tilde N(\beta_{1:2})}\right),\frac{1}{2}\log \frac{\beta_2^2P_2(1+P_R)}{P_R}\}\right\}
\end{IEEEeqnarray*}
where
\begin{IEEEeqnarray*}{rCl}
\tilde N(\beta_{1:2}):=\frac{P_1P_2(a_1\beta_1-a_2\beta_2)^2+(a_1\beta_1)^2P_1+(a_2\beta_2)^2P_2}{P_1+P_2+1}
\end{IEEEeqnarray*}
for any positive $\beta_1,\beta_2$ satisfying $\max\{\beta_1^2P_1,\beta_2^2P_2\}\leq P_R$.  Figure \ref{fig:TWR} shows the achievable rate region. 
\begin{figure}[h!]
\centering
\includegraphics[scale=0.48]{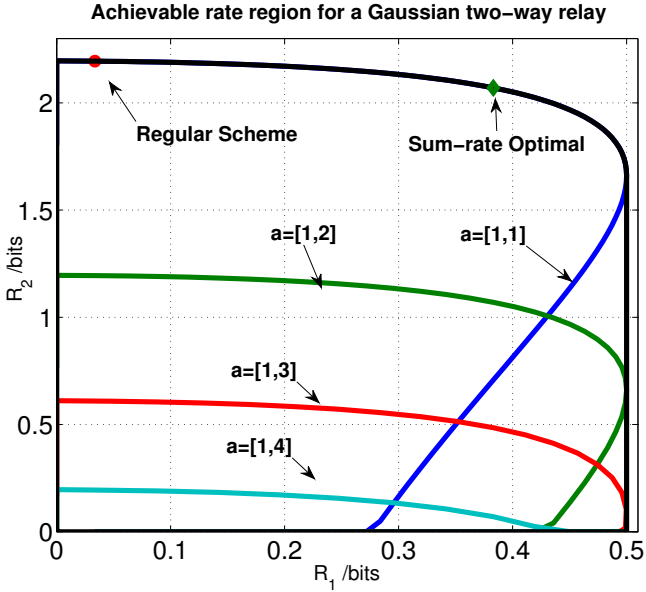}
\caption{Achievable rate region for the Gaussian two-way relay in Figure \ref{fig:system_TWR} with unequal power constraints $P_1=1$, $P_2=20$ and equal channel gain $\ve h=[1,1]$. The relay has power $P_R=20$.  Color curves show different achievable rate region when the relay decodes different linear functions as marked in the plot. The red dot denotes the achievable rate pair given in (\ref{eq:TWR_regular}) when relay decodes $\ve t_1+\ve t_2$ using regular compute-and-forward (other function will give worse rate pair). Notice this point is not sum-rate optimal. The achievable rate region given by the black convex hull is strictly larger than the regular scheme since the CSI can be used at the transmitters.}
\label{fig:TWR}
\end{figure}

\textbf{Example 3: The MIMO integer-forcing linear receiver.}

We now apply the same idea to the MIMO system with an integer-forcing  linear receiver \cite{Zhan_etal_2010}. We consider a point-to-point MIMO system with channel matrix $\ve H\in\mathbb R^{M\times K}$ which is full rank. It is shown in \cite{Zhan_etal_2010} that the following rate is achievable using integer-forcing receiver
\begin{IEEEeqnarray*}{rCl}
R_{IF}&\leq&\min_{m\in[1:k]}K (-\frac{1}{2}\log\ve a_m^T\ve V\ve D\ve V^T\ve a_m)
\end{IEEEeqnarray*}
for any  full rank integer matrix $\ve A\in\mathbb Z^{K\times K}$ with its $m$-th row as $\ve a_m$ and $\ve V\in\mathbb R^{K\times K}$ is composed of the eigenvectors of $\ve H^T\ve H$. The matrix $\ve D\in\mathbb R^{K\times K}$ is diagonal with element $\ve D_{k,k}=1/(P\lambda_k^2+1)$ and $\lambda_k$ is the $k$-th singular value of $\ve H$.

Applying the modified compute-and-forward to the integer-forcing receiver gives the following result. We note that a similar idea also appears in \cite{ordentlich_precoded_2013} where a pre-coding matrix is used at the encoder.
\begin{theorem}
For a $K\times M$ real MIMO system with full rank channel matrix $
\ve H\in\mathbb R^{M\times K}$, the following rate is achievable using an integer-forcing linear receiver for any $\beta_1,\ldots,\beta_K$
\begin{IEEEeqnarray}{rCl}
R_{mIF}&\leq&\sum_{k=1}^K\min_{m\in[1:K]} \Bigg(-\frac{1}{2}\log\frac{\tilde{ \ve a}_m^T\ve V\ve D\ve V^T\tilde {\ve a}_m}{\beta_k^2}\Bigg)
\label{eq:Rate_IntegerForcing}
\end{IEEEeqnarray}
for any  full rank $\ve A\in\mathbb Z^{K\times K}$ with its $m$-th row being $\ve a_m$. We have $\ve{\tilde a}_m:=[\beta_1a_{m1},...,\beta_Ka_{mK}]$ for $m=1,\ldots, K$ and $\ve V, \ve D$ defined as above.
\end{theorem}

In Figure \ref{fig:Integer_forcing} we compare the achievable rates of two schemes. 

\begin{figure}[h!]
\centering
\includegraphics[scale=0.36]{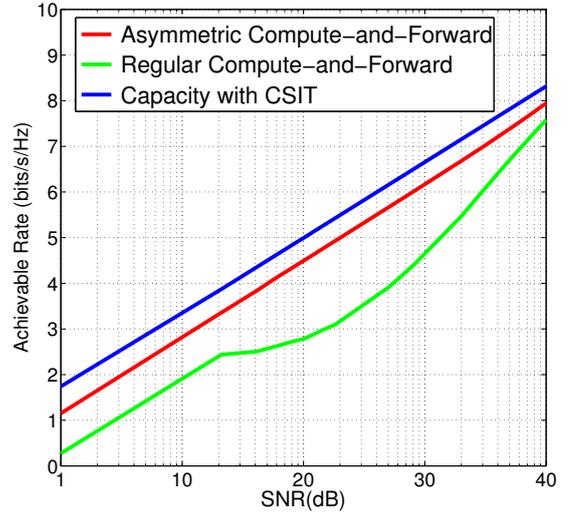}
\caption{Achievable rates for a $2\times 2$ MIMO system $\ve H=[0.7, 1.3; 0.8, 1.5]$.  At $\text{SNR}=40 \text{dB}$, the best coefficients for regular scheme are $\ve a_1=[1, 2]$ and $\ve a_2=[7, 13]$, while for the modified scheme we have the best parameters as $\beta_1=1, \beta_2=4.887, \ve a_1=[8, 3] $ and $\ve a_2=[13, 5]$.  }
\label{fig:Integer_forcing}
\end{figure}

We give another example where the modified scheme performs arbitrarily better than the regular scheme. Consider the $2\times 2$ MIMO channel with channel matrix $\ve H=\begin{bmatrix}
1 &1\\
0 &\epsilon
\end{bmatrix}$ where $0<\epsilon<1$.  It has been shown in \cite[Section V, C] {Zhan_etal_2010} that the achievable rate of integer forcing is upper bounded as $R_{IF}\leq \log(\epsilon^2 P)$ which is of order $O(1)$ if $\epsilon\sim\frac{1}{\sqrt{P}}$ while the joint ML decoding can achieve a rate at least $\frac{1}{2}\log(1+2P)$. With the modified scheme we can show the following result.
\begin{lemma}
For the channel $\ve H$ above, $R_{mIF}$ in (\ref{eq:Rate_IntegerForcing}) scales as $\log P$ for any $\epsilon>0$.
\end{lemma}

To see this, we can show (assuming w. l. o. g. $\beta_1=1$)
\begin{IEEEeqnarray*}{rCl}
R_{mIF}&\geq& \min_{m=1,2}\frac{1}{2}\log^+\left(\frac{P}{a_{m1}^2+(a_{m2}\beta_2-a_{m1})^2\frac{1}{\epsilon^2}}\right)\\
&&+\min_{m=1,2}\frac{1}{2}\log^+\left(\frac{\beta_2^2P}{a_{m1}^2+(a_{m2}\beta_2-a_{m1})^2\frac{1}{\epsilon^2}}\right)
\end{IEEEeqnarray*}
Based on the standard results on simultaneous Diophantine approximation \cite{Cassels_1957}, for any given $a_{m2}$ and $Q>0$ there exists $\beta_2<Q$ and $a_{m1}$ such that $|a_{m2}\beta_2-a_{m1}|<Q^{-1/2}$ for $m=1,2$. Hence the we have the achievable rate
\begin{IEEEeqnarray*}{rCl}
\min_{m=1,2}\frac{1}{2}\log^+\left(\frac{P}{a_{m1}^2+Q^{-1}\frac{1}{\epsilon^2}}\right)+\min_{m=1,2}\frac{1}{2}\log^+\left(\frac{\beta_2^2P}{a_{m1}^2+Q^{-1}\frac{1}{\epsilon^2}}\right)
\end{IEEEeqnarray*}
If we choose $Q\sim\epsilon^{-2}$, and notice that we also have $\beta_2, a_{m1}\sim Q$, then the second term above scales as $\frac{1}{2}\log P$ for $P$ large. Consequently $R_{mIF}$ also scales as $\frac{1}{2}\log P$ for any $\epsilon$, hence can be arbitrarily better than the regular scheme.


\bibliographystyle{IEEEtran}
\bibliography{IEEEabrv,CF_CSIT_IZS}

\end{document}